\documentclass[submission,copyright,creativecommons]{eptcs}
\providecommand{\event}{LSFA 2026} 

\usepackage{iftex}
\usepackage[mathscr]{eucal} 
\usepackage{amsmath} 
\usepackage{amsthm} 
\usepackage{amssymb} 
\usepackage{float} 
\usepackage{xcolor}
\usepackage{color}

\usepackage[all]{xy}                
\usepackage{subfigure}
\usepackage{tabto}                  
\usepackage{enumitem}               

\ifpdf
  \usepackage{underscore}         
  \usepackage[T1]{fontenc}        
\else
  \usepackage{breakurl}           
\fi

\newtheorem{definition}{Definition}
\newtheorem{example}{Example}
\newtheorem{lemma}{Lemma}
\newtheorem{theorem}{Theorem}
\newtheorem{technique}{Technique}

\newcommand\oa{\overline{\alpha}}
\newcommand\pp{\parallel}
\newcommand\acao{\mathcal{A}ct}

\title{Dynamic Logic with Parallel Operator for Verifying Communication Protocols}
\author{
    Luiz C. F. Fernandez
    \institute{
        PESC/COPPE\\
        Universidade Federal do Rio de Janeiro\\
        Rio de Janeiro/RJ, Brazil
    }
    \email{lcfernandez@cos.ufrj.br}
\and
    Mario R. F. Benevides
    \institute{
        Instituto de Computação\\
        Universidade Federal Fluminense\\
        Niterói/RJ, Brazil}
    \email{mario@ic.uff.br}
}

\def\titlerunning{Dynamic Logic with Parallel Operator for Verifying Communication Protocols}
\def\authorrunning{Luiz C. F. Fernandez \& Mario R. F. Benevides}
\begin{document}
\maketitle

\begin{abstract}
    In this paper, we present a dynamic logic with parallel operators for the formal verification of authenticity and safety properties of cryptographic protocols. The logic incorporates communication actions and is specifically designed to reason about protocol executions in adversarial environments. We extend an existing dynamic logic with parallel operators by introducing concepts derived from the Dolev–Yao intruder model. As the underlying logic is completely axiomatizable, we obtain a complete axiomatization for the extended system. Furthermore, we develop a tableau calculus for the proposed logic and prove its termination, soundness, and completeness.
\end{abstract}

\section{Introduction}

The aim of this work is to bring together the encryption/decryption features of Dolev-Yao systems, the communication action of process algebras and the capacity of reasoning about programs of PDL into one framework. This allows for having a unique framework for reasoning about security, communication, synchronisation using the power of dynamic logic.

    The Dolev-Yao model \cite{DOLEV1983} is a seminal work in the area of formal criptography. Here we are most interested in logical approaches to verify authenticity and secrecy in communication protocols. This model uses a deductive approach to prove that the security of a protocol can be broken by a malicious intruder. Over the years, many works have been trying to combine these concepts with logic \cite{COHEN2007,KRAMER2008,BOUREANU2009,benevidesetal2018dymael}.

    Propositional Dynamic Logic (PDL) \cite{FISCHER1979,hkt} is a well-known multi-modal system, that uses concepts involving properties and dynamic behaviors, which permits us to model and reason about actions in programs. The semantics for PDL is based on the notion of \emph{Labelled Transition Systems} (LTS), very similar to regular \emph{Kripke structures} \cite{KRIPKE1959}.

    With Process Algebras \cite{milner89,FOKKINK2000,BERGSTRA2001} we also are able to work with LTS and specify communication. These models also consider concurrency and interaction
between processes, besides the notion of \emph{bisimulation} for the equivalence of processes.

    There are many different proof calculi, e.g., resolution, natural deduction and tableaux, from different approaches, namely \emph{direct} or \emph{indirect deduction} and \emph{labelled deductive systems}. In the latter, we have \emph{prefixed tableaux} \cite{FITTING1983}, which also has a proof representation similar to Kripke semantics and have been widely implemented for PDL and for modal logics in general \cite{PRATT1978,PRATT1980,DEGIACOMO2000,MASSACCI2000}.

    In the next section we present some basic notions for the background for the rest of the work: the Dolev-Yao model, Propositional Dynamic Logic, Process Calculus and Tableaux Calculus. In section \ref{DDYL} we propose a Dynamic Logic for verifying communication protocols. Finally, we present a tableaux calculus for this logic in section \ref{TCDDYL}, with soundness, completeness and termination proofs, and we conclude with some final remarks in section \ref{conclusion}.

\section{Background}\label{BKG}
    \subsection{Dolev-Yao Model}\label{DYModel}
        Introduced by Dolev and Yao \cite{DOLEV1983} at the time of great discussion about the use of public key encryption in network communication, this work intends to show why a formal model is desirable to deal with security protocols.

        Public key systems are efficient when we have a ``passive'' saboteur (also called eavesdropper, attacker, intruder and so on), who only intercepts the communication and tries to decode the message. But Needham and Schroeder \cite{NEEDHAM1978} already had pointed out that a not well specified protocol permits an ``active'' intruder, one who may fake his identity and manipulate the intercepted message, to succeed.

        \subsubsection{Public key protocols}
            To briefly explain this system \cite{DIFFIE1976,RIVEST1978}, we assume that every user $X$ in the network has an \emph{encryption function} $E_X$, which generates a pair ($X$, $E_X$), which is stored in a secure public directory, and a \emph{decryption function} $D_X$, known only to user $X$. One should notice that the sender's public key is represented, in the message exchange, as a subscript of $E$. The main requirements on the functions above are:

            \begin{itemize}
                \item $D_X (E_X (M)) = M$;

                \item for any user $Y$, knowing $E_X(M)$ and the directory containing all the public pairs does not reveal anything about $M$.
            \end{itemize}

            So, other users can communicate with $X$ by sending an encrypted message $E_X(M)$ and $X$ can decrypt it using $D_X(E_X(M)) = M$, but only $X$ gets $M$, even if $E_X(M)$ is accessible to everyone.

            A message transmitted between two users is denoted by: the sender's name, the text (encrypted) and the receiver's name. One of the basic assumptions in the perfect public key system is that the functions are unbreakable.

            To illustrate intruder's possible behaviours, let's consider the following example.

            \begin{example}\label{ex-dy-mitm}
                In this example, also called \emph{Man-in-the Middle (MITM) attack}, the plaintext is encoded with an encryption function, where the receiver always replies using the sender's public key. Suppose user $A$ wants to send a plaintext $M$ to user $B$:

                {\small

                \begin{enumerate}
                    \item[a)] $A$ tries to send a message $(A, E_B (M), B)$ to $B$ \emph{[Figure \ref{fig-dy-mitm-step-1}]};

                    \item[b)] Intruder Z intercepts the above message and sends message $(Z, E_B (M), B)$ to $B$ \emph{[Figure \ref{fig-dy-mitm-step-2}]};

                    \item[c)] $B$ sends message $(B, E_Z (M), Z)$ to $Z$ \emph{[Figure \ref{fig-dy-mitm-step-3}]};

                    \item[d)] $Z$ decodes $E_Z (M)$ and obtains $M$.
                \end{enumerate}

                }
            \end{example}

            \vspace{-20pt}

            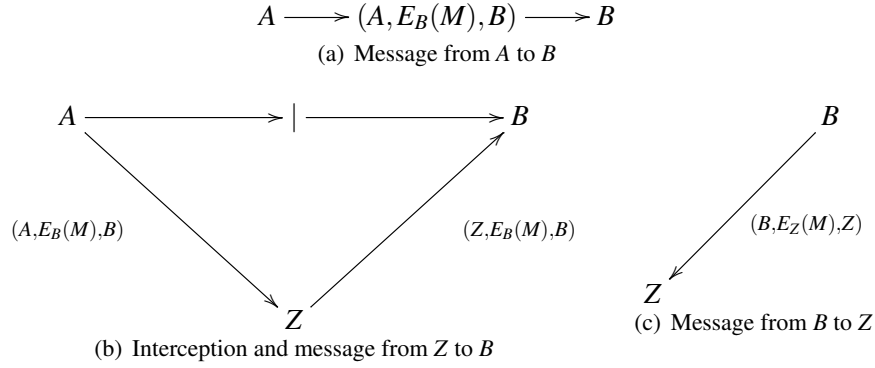
\begin{figure}[H]
                \centering

                \subfigure[Message from $A$ to $B$]{
                    $\xymatrix{A \ar[r] & (A, E_B (M), B) \ar[r] & B}$
                    \label{fig-dy-mitm-step-1}
                }
                \\
                \subfigure[Interception and message from $Z$ to $B$]{
                    $\xymatrix{A \ar[rr] \ar[ddrr] && \mid \ar[rr] && B\\
                    {}_{(A, E_B (M), B)} &&&& {}_{(Z, E_B (M), B)}\\
                    && Z \ar[uurr] &&}$
                    \label{fig-dy-mitm-step-2}
                }
                \quad
                \subfigure[Message from $B$ to $Z$]{
                    $\xymatrix{
                    && B \ar[ddll]^{(B, E_Z (M), Z)}\\
                    &&\\ Z &&
                    }$
                    \label{fig-dy-mitm-step-3}
                }

                \caption{Illustration of Example \ref{ex-dy-mitm}}
            \end{figure}

            \vspace{-40pt}

        \subsubsection{Rules}
            The rules presented below are not formulated in the original paper \cite{DOLEV1983}, but we can easily obtain them from the theory presented there and put them in a new notation. They permit the intruder to make deductions from the intercepted and sent messages.

            Here, we are assuming an enumerable set $\mathscr{K} =\{ k_1, \dots \}$ of keys, a set $T$ containing all the information (messages, keys, etc) that the intruder has and an encryption function $\{\_\}_k$, which encrypts the message $M$ under the key $k$, $\{M\}_k$.

            The \emph{entailment} relation $T \vdash M$ has the intuitive meaning that $M$ can be computed from $T$. This relation is defined inductively, in a natural deduction-like system. Some works have defined a similar notation \cite{ABADI2002,AYALA-RINCON2013}:

            \

            {\small

            \noindent \centerline{\emph{Reflexivity} $\dfrac{M \in T}{T \vdash M}$ \qquad \emph{Encryption} $\dfrac{{T \vdash M} \quad {T \vdash k}}{
            {T \vdash \{M\}_k}}$ \qquad \emph{Decryption} $\dfrac{{T \vdash \{M\}_k} \quad {T \vdash k} } {{T \vdash M}}$}

            \

            \

            \noindent \centerline{\emph{Pair-Composition} $\dfrac{{T \vdash M} \quad {T \vdash N}}{
            {T \vdash (M,N)}}$ \qquad \emph{Pair-Decomposition} $\dfrac{{T \vdash (M,N)}}{{T \vdash M}} \qquad \dfrac{{T \vdash (M,N)}}{{T \vdash N}}$}

            }

    \subsection{Propositional Dynamic Logic}\label{PDL}
        Propositional Dynamic Logic (PDL) was conceived to reason about programs \cite{hkt,goldblatt92,pmy}. Its most common operators are: non-deterministic choice ($\cup$), sequential composition ($;$), iteration ($*$) and test ($?$), this logic is called PDL for regular programs. Its semantics is given by Labelled Transition Systems (LTS), where the transition relation of the LTS is  a binary relation $R_{\pi}$, for each program $\pi$. The sequential composition, non-deterministic choice and iteration operators are defined as the composition, union and transitive reflexive closure of relations respectively: $$R_{\pi_1 ; \pi_2} = R_{\pi_1} \circ R_{\pi_2} \quad R_{\pi_1 \cup \pi_2} = R_{\pi_1} \cup~ R_{\pi_2} \quad R_{\pi^*} = (R_{\pi})^*$$

        In this section, we present the syntax and semantics of PDL.

        \begin{definition}\label{def-langpdl}
            The PDL alphabet consists of a set $\Phi$ of countably many propositional symbols, the propositional constant $\top$ (true), a set ${\cal P}$ of countably many basic programs (or action names), the Boolean connectives $\neg$ and $\land$, the program constructors $;$ (sequential composition), $\cup$ (non-deterministic choice), $?$ (test) and $\phantom{}^*$ (iteration) and a modality $\langle \pi \rangle$ for every program $\pi$. The formulas are defined by the following BNF:

            $$\varphi ::= p \mid \top \mid \neg \varphi \mid (\varphi_1 \wedge \varphi_2) \mid \langle \pi \rangle \varphi , \text{ with}$$
            $$\pi ::= \alpha \mid (\pi_1;\pi_2) \mid (\pi_1 \cup \pi_2) \mid  \varphi? \mid \pi^* ,$$
            \noindent where $p$ ranges over $\Phi$ and $\alpha$ ranges over ${\cal P}$.
        \end{definition}

        In all the logics that appear in this paper, we use the standard abbreviations: $\bot := \neg \top$, $\varphi \lor \phi := \neg ( \neg \varphi \land \neg \phi)$, $\varphi \rightarrow \phi := \neg( \varphi \land \neg \phi)$ and $[\pi] \varphi := \neg \langle \pi \rangle \neg \varphi$.

        \begin{definition}\label{framepdl}
            A \emph{model} for PDL is a tuple $\mathcal{M}= (W, {\bf R}, {\bf V})$ where:

            \begin{itemize}
                \item $W$ is a non-empty set of states;
                \item ${\bf R} = \{R_{\alpha} \mid \alpha \in {\cal N} \}$, $R_\alpha$ are binary relations over $W$, for each basic program $\alpha \in {\cal P}$;
                \item We can  define a binary relation $R_{\pi}$ by mutual induction (Definition \ref{def-satpdl}), for each non-basic program $\pi$, as follows:

                \NumTabs{9}
                \begin{itemize}
                    \item $R_{\pi_1 ; \pi_2}$ \tab $ = R_{\pi_1} \circ R_{\pi_2}$,
                    \item $R_{\pi_1 \cup \pi_2}$ \tab $ = R_{\pi_1} \cup R_{\pi_2}$,
                    \item $R_{\varphi?}$ \tab $ = \{(w,w) \mid {\cal M},w \Vdash \varphi \}$,
                    \item $R_{\pi^*}$ \tab $ = R_{\pi}^*$, where $R_{\pi}^*$ is the reflexive transitive closure of $R_{\pi}$.
                \end{itemize}

                \item ${\bf V}$ is a valuation function ${\bf V} : \Phi \to 2^W$.
            \end{itemize}

            We say that ${\cal F}= (W, {\bf R})$ is a PDL frame.
        \end{definition}

        \begin{definition}\label{def-satpdl}
            Let $\mathcal{M}= ({\cal F}, {\bf V})$ be a model. The notion of \emph{satisfaction} of a formula $\varphi$ in a model $\mathcal{M}$ at a state $w$, notation $\mathcal{M},w \Vdash \varphi$, can be  defined by a mutual induction (Definition \ref{framepdl}) as follows:

            \begin{itemize}
                \item $\mathcal{M},w \Vdash p$ iff $w \in {\bf V}(p)$;
                \item $\mathcal{M},w \Vdash \top$ always;
                \item $\mathcal{M},w \Vdash \neg \varphi$ iff $\mathcal{M},w \not\Vdash \varphi$;
                \item $\mathcal{M},w \Vdash \varphi_{1} \wedge \varphi_{2}$ iff $\mathcal{M},w \Vdash \varphi_{1}$ and $\mathcal{M},w \Vdash \varphi_{2}$;
                \item $\mathcal{M},w \Vdash \langle \pi \rangle \varphi$ iff there is $w' \in W$ such that $w R_{\pi} w'$ and $\mathcal{M},w' \Vdash \varphi$.
            \end{itemize}
        \end{definition}

    \subsection{Process Calculus}
        In this section, we propose a small process (program) calculus for a subset of PDL programs extended with a parallel composition operator.

        We also do not make any distinction between programs and processes. In this work,  processes and programs are used interchangeably.

        Let ${\cal N} = \{x,y,z,\ldots\} $  be a set of names or input actions, denoted by $\alpha, \beta, ...$. We have the special action $END$, which corresponds to the program that is incapable of performing any running action, but it is capable of successfully finishing.

        The set of output actions called co-names is $ \overline{\cal N} =\{ \overline{x}, \overline{y},...\}$ such that $x \in \cal{N}$ iff $\overline{x} \in  \overline{\cal N} $. There is a special action called, silent action $\tau(.)$,  denoting internal communication action. The set of all possible actions is defined as $\acao = {\cal N} \cup \overline{\cal N} \cup \{ \tau \} \cup \{ END \}$. The language is defined as follows:
        \smallskip

        $\pi ::= \alpha.\pi \mid END \mid (\pi_1;\pi_2) \mid (\pi_1 + \pi_2)  \mid (\pi_1 \pp \pi_2),~ where~p \in \Phi ~ and ~\alpha \in \acao$.
        \smallskip

        The semantics of our process calculus is given by the transition rules, labelled by programs, presented in the Table \ref{tab:semccs-par}:


        \begin{table}
            \centering\Large
            \begin{tabular}{|c|c|c|c|}
                \hline
                $\frac{}{\alpha.\pi \stackrel{\alpha}{\rightarrow} \pi}$ &
                $\frac{\pi \stackrel{\alpha}{\rightarrow} \pi'} {\pi ; \gamma \stackrel{\alpha}{\rightarrow} \pi' ; \gamma}$ &
                $\frac{}{END \stackrel{END}{\rightarrow} \surd}$ &
                \\
                \hline
                $\frac{\pi \stackrel{\alpha}{\rightarrow} \pi'} {\pi + \gamma \stackrel{\alpha}{\rightarrow} \pi'}$ &
                $\frac{\gamma \stackrel{\beta}{\rightarrow} \gamma'} {\pi + \gamma \stackrel{\beta}{\rightarrow} \gamma'}$ &
                $\frac{\pi \stackrel{END}{\rightarrow} \surd} {\pi + \gamma \stackrel{END}{\rightarrow} \surd}$ &
                $\frac{\gamma \stackrel{END}{\rightarrow} \surd} {\pi + \gamma \stackrel{END}{\rightarrow} \surd}$
                \\
                \hline
                $\frac{\pi \stackrel{\alpha}{\rightarrow} \pi'}{\pi \pp \gamma \stackrel{\alpha}{\rightarrow} \pi' \pp \gamma}$ &
                $\frac{\gamma \stackrel{\beta}{\rightarrow} \gamma'}{\pi \pp \gamma \stackrel{\beta}{\rightarrow} \pi \pp \gamma'}$ &
                $\frac{\pi \stackrel{\alpha}{\rightarrow} \pi', \gamma \stackrel{\overline{\alpha}}{\rightarrow} \gamma'}{\pi \pp \gamma \stackrel{\tau}{\rightarrow} \pi' \pp \gamma'}$ &
                $\frac{\pi \stackrel{END}{\rightarrow} \surd, \gamma \stackrel{END}{\rightarrow} \surd}{\pi;\gamma \stackrel{END}{\rightarrow} \surd}$
                \\
                \hline
                $\frac{\pi \stackrel{END}{\rightarrow} \surd, \gamma \stackrel{\alpha}{\rightarrow} \gamma'} {\pi ; \gamma \stackrel{\alpha}{\rightarrow} \gamma'}$ &
                &
                &
                $\frac{\pi \stackrel{END}{\rightarrow} \surd, \gamma \stackrel{END}{\rightarrow} \surd}{\pi \pp \gamma \stackrel{END}{\rightarrow} \surd}$
                \\
                \hline
            \end{tabular}
            \caption{Transition Relation}
            \label{tab:semccs-par}
        \end{table}


        The \emph{sequential composition} operator $;$ denotes that the process will first execute $\pi_1$ and then behave as $\pi_2$. The \emph{summation} (or \emph{non-deterministic choice}) operator $+$ denotes that the process will make a non-deterministic choice to behave as either $\pi_1$ or $\pi_2$. The \emph{parallel composition} operator $\pp$ denotes that the processes $\eta_1, \dots, \eta_n$, performed by agents $1,\dots,n$ respectively, may proceed independently or may communicate through a common channel.

The symbol  $\surd$ is used to express {\it successful termination}.      $\pi \stackrel{\alpha}{\rightarrow} \pi'$ express that the process $\pi$ can perform the action $\alpha$ and after that behave as $\pi'$. $\pi \stackrel{END}{\rightarrow} \surd$ express that the process $\pi$ successfully finishes after performing the action $END$. A process  finishes when there is no possible action left for it to perform.  When a process finishes inside a parallel composition, sequential composition or non-deterministic choice one writes $\pi$ instead of $\pi\pp \surd$, $\pi ; \surd$ and $\pi + \surd$. One uses $\surd$ instead of  $\surd \pp \surd$.

        \subsubsection{Bisimulation}
            The concept of bisimulation is a key notion in any process algebra. It is an equivalence relation between processes which have mutually similar behaviour. The intuition is that two bisimilar processes cannot be distinguished by an external observer. The use of the notion of bisimulation allows one to transform any process to an equivalent one that is a summation of all their possible actions. That is what the Expansion Law (Theorem \ref{teo:EL}) states.


            \begin{definition}[\cite{milner89}](Bisimulation)\label{bisimula}

                \begin{itemize}
                    \item Let $\Pi$ be the set of all programs. A set $Z \subseteq \Pi \times \Pi$ is a \emph{strong bisimulation} if $(\pi_1,\pi_2) \in Z$ implies the following for all $\alpha \in \acao$:

                    \begin{itemize}
                        \item If $\pi_1 \stackrel{\alpha}{\rightarrow} \pi_1'$, then there is $\pi_2' \in \Pi$ such that $\pi_2 \stackrel{\alpha}{\rightarrow} \pi_2'$ and $(\pi_1',\pi_2') \in Z$;
                        \item If $\pi_2 \stackrel{\alpha}{\rightarrow} \pi_2'$, then there is $\pi_1' \in \Pi$ such that $\pi_1 \stackrel{\alpha}{\rightarrow} \pi_1'$ and $(\pi_1',\pi_2') \in Z$;
                        \item $\pi_1 \stackrel{\alpha}{\rightarrow} \surd$ if and only if $\pi_2 \stackrel{\alpha}{\rightarrow} \surd$.
                    \end{itemize}

                    \item Two process $\pi$ and $\pi'$ are \emph{strongly bisimilar} (or simply \emph{bisimilar}), denoted by $\pi \simeq \pi'$, if there is a strong bisimulation $Z$ such that $(\pi,\pi') \in Z$.
                \end{itemize}
            \end{definition}

            We introduce the Expansion Law, which is very important in the definition of the semantics of our logic and its axiomatization. We present a particular case of the Expansion Law, which is suited to our needs. We only use it for the case of parallel composition operator. The most general case of the Expansion Law is presented in \cite{milner89}.

            \begin{theorem}[Expansion Law (EL)]\label{teo:EL}
                Let $\pi = \pi_1 \pp \pi_2$. Then
                    $$
                    \pi \simeq \sum_{\pi_1 \stackrel{\alpha}{\rightarrow} \pi_1'} \alpha.(\pi_1' \pp \pi_2) + \sum_{\pi_2 \stackrel{\beta}{\rightarrow} \pi_2'} \beta.(\pi_1 \pp \pi_2') + \sum_{R \in A_{\tau}} \tau . R,
                    $$
                where $A_{\tau} = \{ (\pi_1' \pp \pi_2') : \pi_1 \stackrel{\alpha}{\rightarrow} \pi_1' \mbox{ and } \pi_2
                \stackrel{\overline{\alpha}}{\rightarrow} \pi_2', \mbox{ for some } \alpha \in {\cal N} \} \cup \{ (\pi_1' \pp \pi_2') : \pi_1
                \stackrel{\overline{\alpha}}{\rightarrow} \pi_1' \text{  and } \pi_2 \stackrel{\alpha}{\rightarrow} \pi_2', \text{ for some } \alpha \in {\cal N} \}$.
                If $A_{\tau} = \emptyset$, then $\sum_{R \in A_{\tau}} \tau . R $ must be removed from the equation.
                We denote the right side of this bisimilarity by $Exp(\pi)$.
            \end{theorem}

            The Expansion Law is a very useful property in process algebras. Its intuition is that processes can be rewritten as a summation of all their possible actions. Suppose we have processes $A \stackrel{def}{=} x.A'$ and $B \stackrel{def}{=} \overline{x}.B'$, then the process $(A \pp B)$ is equivalent, using the Expansion Law, to: $$(A \pp B) \simeq x.(A' \pp B) + \overline{x}.(A \pp B') + \tau.(A' \pp B').$$

    \subsection{Tableaux Calculus}\label{Tableaux}
        In this section, we present the tableaux method. This proof procedure already have been provided for PDL \cite{PRATT1978,PRATT1980}. The following definitions and rules are based on proposals for some extensions and for modal logics \cite{DEGIACOMO2000,MASSACCI2000}.

        \begin{definition}
            \emph{Prefixed formulas} are pairs $\langle \sigma : \varphi \rangle$ where $\varphi$ is a formula and $\sigma$ is defined as $\sigma ::= 1 \mid \sigma.a.n$, i.e., an alternating sequence of integers and atomic programs, starting from the initial state 1 and reaching the state where $\varphi$ holds. We say that $\sigma$ is a \emph{prefix}.
        \end{definition}

        \begin{definition}[\cite{FITTING1983,GORE1999,MASSACCI1994}]
            A \emph{tableau} $\mathscr{T}$ is a rooted tree where nodes are labelled with prefixed formulas, a \emph{branch} $\theta$ is a path from the root to a leaf (intuitively, this is a model for the initial formula) and a \emph{segment} ${\cal S}$ is a path from the root to a node of the tree.
        \end{definition}

        \begin{definition}
            A prefix is \emph{present} in a segment if there is a prefixed formula with that prefix already in the segment, and it is \emph{new} otherwise.
            \vspace{-10pt}
        \end{definition}

        \subsubsection{Rules}
            Here we present the classical and PDL tableau rules, including for an atomic program $A$:

            \

            {\small

            \centerline{
                R$_\land$ $\dfrac{\sigma: \varphi \land \psi}{\genfrac{}{}{0pt}{}{}{\genfrac{}{}{0pt}{0}{\sigma: \varphi}{\sigma: \psi}}}$
                \qquad
                R$_\land^\neg$ $\dfrac{\sigma: \neg (\varphi \land \psi)}{\sigma: \neg \varphi \quad \sigma: \neg \psi}$
                \qquad
                R$_\text{Dneg}$ $\dfrac{\sigma: \neg \neg \varphi}{\sigma: \varphi}$
                \qquad
                R$_\text{Seq}$ $\dfrac{\sigma: \langle \pi_1 ; \pi_2 \rangle \varphi}{\sigma: \langle \pi_1 \rangle \langle \pi_2 \rangle \varphi}$

            }

            \

            \

            \centerline{
                R$_\text{Seq}^\neg$ $\dfrac{\sigma: \neg \langle \pi_1 ; \pi_2 \rangle \varphi}{\sigma: \neg \langle \pi_1 \rangle \langle \pi_2 \rangle \varphi}$
                \qquad
                R$_\text{Test}$ $\dfrac{\sigma: \langle \varphi ? \rangle \psi}{\genfrac{}{}{0pt}{}{}{\genfrac{}{}{0pt}{0}{\sigma: \varphi}{\sigma: \psi}}}$
                \qquad
                R$_\text{Test}^\neg$ $\dfrac{\sigma: \neg \langle \varphi ? \rangle \psi}{\sigma: \neg \varphi \quad \sigma: \neg \psi}$
                \qquad
            }

            \

            \

            \centerline{
                R$_\text{Choice}$ $\dfrac{\sigma: \langle \pi_1 + \pi_2 \rangle \varphi}{\sigma: \langle \pi_1 \rangle \varphi \quad \sigma: \langle \pi_2 \rangle \varphi}$
                \qquad
                R$_\text{Choice}^\neg$ $\dfrac{\sigma: \neg \langle \pi_1 + \pi_2 \rangle \varphi}{\genfrac{}{}{0pt}{}{}{\genfrac{}{}{0pt}{0}{\sigma: \neg \langle \pi_1 \rangle \varphi}{\sigma: \neg \langle \pi_2 \rangle \varphi}}}$
            }

            \

            \

            \centerline{
                R$_{\langle A \rangle}$ $\dfrac{\sigma: \ \langle A \rangle \varphi}{\sigma.A.n: \ \varphi}$, with $\sigma.A.n$ new in the branch
            }

            \

            \

            \centerline{
                R$_{\langle A \rangle}^\neg$ $\dfrac{\sigma: \ \neg \langle A \rangle \varphi}{\sigma.A.n: \ \neg \varphi}$, with $\sigma.A.n$ already present in the branch
            }

            }

            \

            We omit the rules for $\phantom{}^*$-iteration operator since our proposal does not include such operator. If a segment terminates into a branching due to R$_\wedge^\neg$, R$_\text{Test}^\neg$ or R$_\text{Choice}$, we denote the left-hand extension of ${\cal S}$ with ${\cal S}_l$ and the right-hand extension with ${\cal S}_r$.

            Given this set of rules, for ${\cal S}$ is a segment (possibly a branch) of a tableau, the notation ${\cal S}/ \sigma$ stands for the set of prefixed formulas in ${\cal S}$ labelled with the prefix $\sigma$: ${\cal S}/ \sigma = \{\varphi ~|~ \langle \sigma : \varphi \rangle \in {\cal S}\}$.

            \begin{definition}\label{def-tableaux-branch-reduced}
                A prefix $\sigma$ is \emph{reduced} in ${\cal S}$ if $\langle A \rangle$-rules are the only rules not yet applied to formulas of ${\cal S}/ \sigma$ and it is \emph{fully reduced} if all rules have been applied.
            \end{definition}

            \begin{definition}
                A prefix $\sigma$ in the segment ${\cal S}$ is a copy of a prefix $\sigma_0$ in ${\cal S}_0$ if $\cal S$ / $\sigma = $ ${\cal S}_0$ / $\sigma_0$ and both have the same form $\sigma'.A.n$ and $\sigma_0'.A.m$ for the same atomic program $A$.
            \end{definition}

             The pair $\langle \sigma_0, {\cal S}_0 \rangle$ is shorter than $\langle \sigma, \theta \rangle$ if $\sigma_0$ is a proper initial subsequence of $\sigma$ and ${\cal S}_0$ is an initial subsegment of $\theta$ or $\sigma = \sigma_0$ and ${\cal S}_0$ is a proper initial subsegment of $\theta$.
             
            \begin{definition}
                A branch $\theta$ is $\pi$-\emph{completed} if:
                \begin{enumerate}
                    \item all prefixes are reduced;
                    \item for every $\sigma$ which is not fully reduced there is a pair $\langle \sigma_0, {\cal S}_0 \rangle$ shorter than $\langle \sigma, \theta \rangle$ such that $\sigma_0$ is fully reduced in the segment ${\cal S}_0$ and $\sigma$ is a copy of $\sigma_0$ in $\theta$.
                \end{enumerate}
            \end{definition}

            \begin{definition}
                A branch $\theta$ is \emph{contradictory} iff, considering some $\sigma$ and some $P$, it contains both $\sigma : P$ and $\sigma : \neg P$.
            \end{definition}

            \begin{definition}
                A tableau is \emph{closed} if all branches are contradictory and it is \emph{open} if at least one branch is open ($\pi$-completed and non-contradictory).
            \end{definition}

            \begin{definition}
                A \emph{tableau validity proof} for the formula $\varphi$, given other formulas as premises, is the closed tableau starting with these premises and $\neg \varphi$.
            \end{definition}

\section{Dynamic Dolev-Yao Logic}\label{DDYL}
    In this section we present the language, semantics and axiomatization of our Dynamic Dolev-Yao Logic (DDYL). We refrain from using the iteration operator, as its interaction with the composition operator leads to a substantial increase in computational complexity (see \cite{tcs/Benevides17}).

    \subsection{Language}\label{langddyl}
        In the language of DDYL, formulas are built from expressions and not only from propositional symbols. Intuitively, an expression is any piece of information that can be encrypted, decrypted or concatenated in order to be communicated.

        \begin{definition}
            The language of DDYL consists of an enumerable set $\Phi$ of propositional symbols, a finite set $\mathscr{A}$ of agents, a finite set of keys $\mathscr{K} = \{k_a, \bar{k_a}, \cdots \}$, two for each agent, one public and one private key for each agent $a$, the Boolean connectives $\neg$ and $\land$ and a modality $\langle \pi \rangle$, for each protocol $\pi$. The expressions and formulas are defined by the following BNF:

            $$E ::=  p \mid k \mid (E_1, E_2) \mid \{E\}_k,$$

            where $k \in \mathscr{K}$ and $p \in \Phi$.

            $$\varphi ::= e \mid \top \mid \neg \varphi \mid \varphi_1 \wedge \varphi_2 \mid \langle \pi \rangle \varphi,$$

            where $e \in E$ and $\pi$ is a protocol defined as follows:

            $$\pi : : = \alpha(m) \mid \pi_1 + \pi_2 \mid \pi_1 ; \pi_2 \mid \varphi? \mid \pi_1 \pp \pi_2,$$

            where $\alpha \in \acao$ is a communication action/port.

            A message $m$ is any expression  $e$, where $e \in E$.

            We have \textbf{input actions/ports} and \textbf{output actions/ports}. We use the convention that $\alpha(m)$ is an input action receiving message $m$ and $\oa(m)$ is its correspondent output action sending message $m$. Their intuitive meaning is:

            \begin{itemize}
                \item $\oa(m)$ - ``a message $m$ is sent on communication port $\alpha$'';

                \item $\alpha(m)$ - ``a message $m$ is received on communication port $\alpha$''.
            \end{itemize}

            We also have the set ${\cal A} = \{\tau (m_1), \tau (m_2) \cdots \}$ of joint communication actions, one for each pair $(\alpha(m), \oa(m))$.

        \end{definition}

    \subsection{Semantics}\label{semanticsddyl}
        This section presents the notions of models and satisfaction.

        \begin{definition}\label{frameddyl}
            A \emph{model} for DDYL is the tuple $\mathcal{M}= (W, {\bf R}, {\bf V})$ where:

            \begin{itemize}
                \item $W$ is a non-empty set of states;
                \item ${\bf R} = \{R_{\alpha} \mid \alpha \in \acao \}$, $R_a$ are binary relations over $W$, for each basic program $\alpha \in \acao$. For the actions $END$ and $\tau$ the following conditions must be satisfied:
                \begin{enumerate}
                    \item[i.] $R_{END} = \{(w,w) \mid w \in W\}$;
                    \item[ii.] $R_{\tau}$ is serial, for all $w \in W$, there exists $w'$, such that $(w,w') \in R_{\tau}$;
                    \item[iii.] If $(w,w') \in R_{\tau(m)}$ then $\mathcal{M}, w' \Vdash m$.
                \end{enumerate}
                \item We can  define a binary relation $R_{\pi}$ by mutual induction (Definition \ref{sat-ddyl}), for each non-basic program $\pi$, as follows:

                \NumTabs{18}
                \begin{itemize}
                    \item $R_{\pi_1 ; \pi_2}$ \tab $ = R_{\pi_1} \circ R_{\pi_2}$,
                    \item $R_{\pi_1 + \pi_2}$ \tab $ = R_{\pi_1} \cup R_{\pi_2}$,
                    \item $R_{\pi_1 \pp \pi_2}$ \tab $ =  R_{ \alpha_1. \gamma_1} \cup ... \cup R_{\alpha_n.\gamma_n}$, where $Exp(\pi_1 \pp \pi_2) = \alpha_1. \gamma_1 + ... + \alpha_n.\gamma_n$ and $\alpha_i \in \acao$, for $1 \leq i \leq n$. \footnote{$Exp(\pi_1 \pp \pi_2)$ is as defined in Theorem \ref{teo:EL}}
                \end{itemize}

                \item ${\bf V}$ is a valuation function $V : E \to 2^W$ satisfying the following conditions for all $m \in E$, $k_a, \bar{k_a} \in \mathscr{K}$ and all agents $a \in \mathscr{A}$:

                \begin{enumerate}
                    \item\label{condition:ddylModelEncription} $V(m) \cap V(k_a) \subseteq V(\{m\}_{k_a})$
                    \item\label{condition:ddylModelDecryption} $V(\{m\}_{k_a}) \cap V(\bar{k_a}) \subseteq V(m)$
                    \item\label{condition:ddylModelPair} $V(m) \cap V(n) = V((m,n))$
                \end{enumerate}
            \end{itemize}
            We say that ${\cal F}= (W, {\bf R})$ is a DDYL frame.
        \end{definition}

        \begin{definition}\label{sat-ddyl}
            Let $\mathcal{M}= ({\cal F}, {\bf V})$ be a model. The notion of \emph{satisfaction} of a formula $\varphi$ in a model $\mathcal{M}$ at a state $w$, notation $\mathcal{M},w \Vdash \varphi$, can be defined by a mutual induction (Definition \ref{frameddyl}) as follows:

            \begin{itemize}
                \item $\mathcal{M},w \Vdash \top$ always;
                \item $\mathcal{M},w \Vdash e$ iff $w \in {\bf V}(e)$;
                \item $\mathcal{M},w \Vdash \neg \varphi$ iff $\mathcal{M},w \not\Vdash \varphi$;
                \item $\mathcal{M},w \Vdash \varphi_{1} \wedge \varphi_{2}$ iff $\mathcal{M},w \Vdash \varphi_{1}$ and $\mathcal{M},w \Vdash \varphi_{2}$;
                \item $\mathcal{M},w \Vdash \langle \pi \rangle \varphi$ iff there is $w' \in W$ such that $w R_{\pi} w'$ and $\mathcal{M},w' \Vdash \varphi$.
            \end{itemize}
        \end{definition}

        It is important to notice that we are defining the relation $R_{\pi_1 \pp \pi_2}$ as the union of the relations of each term of its expansion.


    \subsection{Axiomatization}\label{ax-ddyl}
        The axiomatization presented here combines axioms from two logics. One is a Propositional Dynamic Logic with Communication Action and Parallel Operator \cite{jolc14} and the other is a Dolev-Yao Multi-agent Epistemic Logic presented in \cite{benevidesetal2018dymael}. Let $\pi$, $\pi_1$ and $\pi_2$ be processes and $\alpha_i \in \acao$, for $1 \leq i\leq n$:

        \subsubsection{Axioms}
            \begin{enumerate}
                \item {\it All propositional logic tautologies},
                \item $[ \pi ](p \rightarrow q) \rightarrow ([ \pi ]p \rightarrow [ \pi ]q)$,
                \item $[ \pi_{1} ; \pi_{2} ] p \leftrightarrow [ \pi_{1} ] [  \pi_{2} ] p$,
                \item $\langle \pi_{1} + \pi_{2} \rangle p \leftrightarrow (\langle \pi_{1} \rangle p \lor \langle \pi_{2} \rangle p)$,
                \item\label{ax-par} $ \langle \pi_1 \pp \pi_2 \rangle p  \leftrightarrow \langle\alpha_1. \gamma_1 \rangle p \lor ... \lor \langle \alpha_n.\gamma_n \rangle p$, where  $Exp(\pi_1 \pp \pi_2) = \alpha_1. \gamma_1 + ... + \alpha_n.\gamma_n$,
                \item\label{ax-end} $[END] p  \leftrightarrow p$,
                \item\label{ax-ddyl-encryption} $m \land k_a \rightarrow \{m\}_{k_a}$ 
                \item\label{ax-ddyl-decryption} $\{m\}_{k_a} \land \bar{k_a} \rightarrow m$ 
                \item\label{ax-ddyl-pair} $m \land n \leftrightarrow (m, n)$ 
                \item\label{ax-tau} $[\tau(m)] p \to \langle \tau(m) \rangle p$,
                \item\label{ax-tau-p-m} $[\tau(m)] p \to [\tau(m)] (p\land m)$.
            \end{enumerate}


            \vspace{-15pt}

        \subsubsection{Inference Rules}
            $$\textrm{M.P.} ~ \frac{\varphi, \varphi\rightarrow \psi}{\psi}
                \qquad \textrm{U.G.} ~ \frac{\vdash \varphi}{ \vdash [ \pi ]\varphi}
                \qquad \textrm{SUB.} ~ \frac{\varphi}{\rho \varphi},$$

            where $\rho$ is a map uniformly substituting formulas for propositional variables.

            \

            (PCSub) If $\vdash \langle \pi \rangle p \leftrightarrow \langle \tau \rangle p$, then $\vdash \langle \pi \pp \gamma \rangle p \leftrightarrow \langle \tau \pp \gamma \rangle p$ and $\vdash \langle \gamma \pp \pi \rangle p \leftrightarrow \langle \gamma \pp \tau \rangle p$

            \

            Axioms 1, 2, 3 and 4 and the inference rules M.P., U.G. and SUB. are standard in PDL for regular programs \cite{hkt,goldblatt92,pmy}. Axioms 5 is the Expansion rule. Inference rule PCSub enforces some desirable property of the parallel composition operator\footnote{The rule (PCSub) is not written as an inference rule like (M.P), (U.G.) and (SUB.) only for clarity.}. Axioms 7, 8 and 9 enforce the semantical properties of the valuation function (conditions \ref{condition:ddylModelEncription}, \ref{condition:ddylModelDecryption} and \ref{condition:ddylModelPair} of Definition \ref{frameddyl}).


        \subsubsection{Soundness and Completeness}
            The axiomatization presented here combines axioms from two logics. One is a Propositional Dynamic Logic with Communication Action and Parallel Operator \cite{jolc14} and the other is a Dolev-Yao Multi-agent Epistemic Logic presented in \cite{benevidesetal2018dymael}. Both works prove soundness and completeness for their proposed axiomatic systems. As our axiomatization is a combination of both logics, the proofs of soundness and completeness follow straightforward from them.
            
\begin{theorem} DDYL is sound and complete with respect to the class of  DDYL
\emph{models}.
\begin{proof} The proof o soundness and completeness follow straightforward from \cite{benevidesetal2018dymael} and \cite{jolc14}.

\end{proof}
\end{theorem}

\section{Tableaux Calculus for Dynamic Dolev-Yao Logic}\label{TCDDYL}
    In this section we propose a Tableaux Calculus for Dynamic Dolev-Yao Logic. We extend the method presented in section \ref{Tableaux} by adding some rules.

    \subsection{Rules}\label{TCDDYLRules}
    The following tableaux rules indicate the correspondence between our axioms of \emph{encryption}, \emph{decryption} and \emph{pair composition $\&$ decomposition} and the semantical properties of the valuation function ${\bf V}$ (conditions \ref{condition:ddylModelEncription}, \ref{condition:ddylModelDecryption} and \ref{condition:ddylModelPair} of Definition \ref{frameddyl}):

    {\small

    \

    \centerline{
        R$_\text{Dec}$ $\dfrac{\genfrac{}{}{0pt}{}{}{\genfrac{}{}{0pt}{0}{\sigma: \{m\}_{k_a}}{\sigma: \bar{k_a}}}}{\sigma: m}$
        \quad
        R$_\text{Enc}^\neg$ $\dfrac{\sigma: \neg \{m\}_{k_a}}{\sigma: \neg m \quad \sigma: \neg k_a}$
        \quad
        R$_\text{Pair}$ $\dfrac{\sigma: (m, n)}{\genfrac{}{}{0pt}{}{}{\genfrac{}{}{0pt}{0}{\sigma: m}{\sigma: n}}}$
        \quad
        R$_\text{Pair}^\neg$ $\dfrac{\sigma: \neg (m, n)}{\sigma: \neg m \quad \sigma: \neg n}$
    }

    }

    \

    where $m, \{m\}_{k_a}, n, (m,n) \in E$ and $k_a, \overline{k}_a \in \mathscr{K}$. For the parallel operator:

    {\small

    \

    \centerline{
        R$_\text{Par}$ $\dfrac{\sigma:\langle \pi \pp \pi' \rangle \varphi}{{\sigma: \langle  \pi_1  \rangle \varphi}\quad {\cdots} \quad {\sigma: \langle  \pi_n  \rangle \varphi} }$
        \quad
        R$_\text{Par}^\neg$ $\displaystyle{{\sigma: \neg \langle \pi \pp \pi' \rangle \varphi} \over \displaystyle{{\sigma: \neg \langle \pi_1 \rangle \varphi} \atop {\displaystyle{{\vdots} \atop {\sigma: \neg\langle \pi_n \rangle \varphi}}}}}$
    }

    \

    }

    with $Exp(\pi \pp \pi') = \pi_1 + \pi_1 + \ldots +\pi_n$. Finally, to deal with communication actions, we also add the following rules:

    \

    {\small

    \centerline{
        R$_\text{D}$ $\dfrac{\sigma: \neg \langle \tau(m) \rangle \varphi}{\sigma: \langle \tau(m) \rangle \neg \varphi}$
        \qquad
        R$_\text{END}$ $\dfrac{\sigma: \langle END \rangle \varphi}{\sigma: \varphi}$
    }

    \

    \

    \centerline{
        R$_{\langle \tau \rangle}$ $\dfrac{\sigma: \ \langle \tau(m) \rangle \varphi}{\sigma.\tau(m).n: \ m \land \varphi}$, with $\sigma.\tau(m).n$ new in the branch
    }

    }

    \

    \

    {\small

    \centerline{
        R$_{\langle \tau \rangle}^\neg$ $\dfrac{\sigma: \ \neg \langle \tau(m) \rangle \varphi}{\sigma.\tau(m).n: \ m \land \neg \varphi}$, with $\sigma.\tau(m).n$ already present in the branch
    }

    }

    \

    where $\tau(m) \in {\cal A}$.

    \subsection{Soundness}
        The soundness proof for our method is similar to the one presented in \cite{DEGIACOMO2000}. We need some definitions to follow some established steps \cite{DEGIACOMO2000,FITTING1983,MASSACCI1994}:

        \begin{definition}\label{def:tableauxProperties}
            Let $\Gamma$ be a set of formulas:

            \begin{enumerate}
                \item we denote $\mathcal{M}, s \Vdash \Gamma$ to represent $\mathcal{M},s \Vdash \varphi$, for all $\varphi \in \Gamma$;

                \item\label{property:tableauxModel} we say $\Gamma$ is satisfiable if there exists a model $\mathcal{M}$ and some possible state $s \in W$ such that $\mathcal{M},s \Vdash \Gamma$;

                \item\label{property:tableauxBranch} a tableau branch is satisfiable if the set of all its formulas is satisfiable. A tableau is satisfiable if at least one branch is satisfiable.
            \end{enumerate}
        \end{definition}

        \begin{lemma}
            The rules of the tableaux method preserve satisfiability. That is, if a tableau $\mathscr{T}$ is satisfiable then the tableau resulting from the application of a rule to $\mathscr{T}$ is satisfiable.
        \end{lemma}

        \begin{proof}
            Let $\mathscr{T}$ be a satisfiable tableau. By property \ref{property:tableauxBranch} of Definition \ref{def:tableauxProperties}, $\mathscr{T}$ has at least one satisfiable branch, although it could have unsatisfiable ones. So, either the rule is applied to a satisfiable branch or to an unsatisfiable one.

            \textbf{First case:} if the rule is applied to an unsatisfiable branch, each originally satisfiable branch remains unchanged. Therefore, the tableau resulting from the application of a rule is satisfiable.

            \textbf{Second case:} if the rule is applied to a satisfiable branch $\theta$, which consists of a set of formulas $\Gamma$ and some specific formulas $\gamma$ and $\delta$ which the rule is applied. As $\theta$ is satisfiable, by property \ref{property:tableauxModel} of Definition \ref{def:tableauxProperties}, there exists a model $\mathcal{M}$ and a possible state $s \in W$ such that $\mathcal{M},s \Vdash \Gamma$, in particular, $\mathcal{M},s \Vdash \gamma$ and $\mathcal{M},s \Vdash \delta$. Let's $\theta'$ be the new branch obtained by the application of an inference rule to $\theta$. We have the following cases for each possible structure of $\gamma$ or $\delta$:

            \begin{itemize}
                \item for $\gamma$ or $\delta$ of type $\varphi \land \psi$, $\neg (\varphi \land \psi)$ or $\neg \neg \varphi$, and rules R$_\text{Seq}$, $_\text{Seq}^\neg$, R$_\text{Test}$, R$_\text{Test}^\neg$R$_\text{Choice}$, R$_\text{Choice}^\neg$, R$_{\langle A \rangle}$, R$_{\langle A \rangle}^\neg$, the proof can be found in tableaux for modal logics literature \cite{FITTING1983,GORE1999,MASSACCI1994}.

                \item R$_\text{Dec}$: for $\gamma$ of type $m$ and for $\delta$ of type $k$, where $m, k \in \Gamma$, since $\mathcal{M},s \Vdash m$ and $\mathcal{M},s \Vdash k$, that is, $\mathcal{M},s \Vdash m \land k$, by the soundness of axiom \ref{ax-ddyl-encryption} of section \ref{ax-ddyl}, we have $\mathcal{M},s \Vdash \{m\}_k$. Therefore, $\theta'$ is satisfiable.

                \item R$_\text{Enc}^\neg$: for $\gamma$ or $\delta$ of type $\neg \{m\}_k$ and $\mathcal{M},s \Vdash \neg  \{m\}_k$. By the contrapositive of axiom \ref{ax-ddyl-encryption} of section \ref{ax-ddyl} and its soundness, we have $\mathcal{M},s \Vdash \neg (m \land k)$ and also $\mathcal{M},s \Vdash \neg m \lor  \neg k$. Suppose $\mathcal{M},s \Vdash \neg m $, then $\theta'$ is satisfiable. Suppose $\mathcal{M},s \Vdash \neg k $, then $\theta'$ is also satisfiable. Therefore, $\theta'$ is satisfiable.

                \item the cases for rules R$_\text{Pair}$ and R$_\text{Pair}^\neg$ are analogous to the cases for rules R$_\text{Dec}$ and R$_\text{Enc}^\neg$, respectively, but using axiom \ref{ax-ddyl-pair}.

                \item R$_\text{Par}$ and R$_\text{Par}^\neg$: follows straightforward from the soundness of axiom \ref{ax-par} \cite{jolc14}.

                \item R$_\text{END}$ and R$_\text{D}$: follows from the soundness of axioms \ref{ax-end} and \ref{ax-tau}, and the conditions i and ii in Definition \ref{frameddyl}, respectively.

                \item R$_{\langle \tau \rangle}$ and R$_{\langle \tau \rangle}^\neg$: follows from the soundness of axiom \ref{ax-tau-p-m} and the condition iii in Definition \ref{frameddyl}.
            \end{itemize}
        \end{proof}

        The soundness of our tableaux method follows straightforward from the above lemma. If a formula $\neg \alpha$ has a closed tableau, then it is unsatisfiable. Therefore $\alpha$ must be a valid formula.

    \subsection{Completeness}
        The completeness proof for our method is inpired  by \cite{DEGIACOMO2000}, following some established steps \cite{DEGIACOMO2000,FITTING1983,GORE1999,MASSACCI1994}. First, we need some definitions:

        \begin{definition}\label{def:formulaTypes}
            Formulas of the form $X \land Y$, $\neg \neg X$, $(m, n)$ or occurrences of $m$ and $k$ are called type-$\alpha$ formulas, while every formulas of the form $\neg (X \land Y)$, $\neg (m, n)$ or $\neg \{m\}_k$ are called type-$\beta$ formulas. The components $\alpha_1$ and $\alpha_2$ from a type-$\alpha$ formula and the components $\beta_1$ and $\beta_2$ from a type-$\beta$ formula are given in tables \ref{table:tableauxComponentsAlpha} and \ref{table:tableauxComponentsBeta}:

            \begin{table}
            \parbox{.45\linewidth}{
                \centering
                \begin{tabular}{c|c|c}
                    $\alpha$ & $\alpha_1$ & $\alpha_2$ \\
                    \hline
                    $X \land Y$ & $X$ & $Y$ \\
                    $\neg (X \lor Y)$ & $\neg X$ & $\neg Y$ \\
                    $\neg (X \to Y)$ & $X$ & $\neg Y$ \\
                    $\neg \neg X$ & $X$ & $X$ \\
                    $(m,n)$ & $m$ & $n$ \\
                    $\genfrac{}{}{0pt}{}{}{\genfrac{}{}{0pt}{0}{m}{k}}$ & $\{m\}_k$ & $\{m\}_k$ \\
                \end{tabular}
                \caption{Type-$\alpha$ formulas}
                \label{table:tableauxComponentsAlpha} }
           \hfill
            \parbox{.45\linewidth}{
                \centering
                \begin{tabular}{c|c|c}
                    $\beta$ & $\beta_1$ & $\beta_2$ \\
                    \hline
                    $X \lor Y$ & $X$ & $Y$ \\
                    $\neg (X \land Y)$ & $\neg X$ & $\neg Y$ \\
                    $X \to Y$ & $\neg X$ & $Y$ \\
                    $\neg (m,n)$ & $\neg m$ & $\neg n$ \\
                    $\neg \{m\}_k$ & $\neg m$ & $\neg \overline{k}$ \\
                \end{tabular}
                \caption{Type-$\beta$ formulas}
                \label{table:tableauxComponentsBeta} }
                \vspace{-10pt}
            \end{table}
        \end{definition}

        \begin{definition}\label{def:tableauxComplete}
            A branch $\theta$ of a tableau $\sigma$ is called \emph{complete} if it satisfies the following conditions (where $\Sigma$ is a set of formulas of $\theta$, $\gamma$ a specific formula and $\varphi$ and $\psi$ are Type-$\alpha$ and Type-$\beta$ formulas respectively):

            {\small

            \begin{enumerate}
                \item\label{cond:tableauxCompleteAlpha} if $(\sigma, \varphi) \in \Sigma$, then $(\sigma, \varphi_1) \in \Sigma$ and $(\sigma, \varphi_2) \in \Sigma$;

                \item\label{cond:tableauxCompleteBeta} if $(\sigma, \psi) \in \Sigma$, then $(\sigma, \psi_1) \in \Sigma$ or $(\sigma, \psi_2) \in \Sigma$;

                \item\label{cond:tableauxCompleteMod} if $(\sigma, \langle \pi \rangle \gamma) \in \Sigma$, then $(\sigma', \gamma) \in \Sigma$ for every tableau $\sigma'$ that occurs in $\Sigma$ and is accessible from $\sigma$;

                \item\label{cond:tableauxCompleteNMod} if $(\sigma, \neg \langle \pi \rangle \gamma) \in \Sigma$, then $(\sigma', \gamma) \in \Sigma$ for some tableau $\sigma'$ that is accessible from $\sigma$;

                \item\label{cond:tableauxCompleteBranch} every branch of any tableau which is accessible from $\theta$ is complete or closed as well.
            \end{enumerate}

            }
        \end{definition}

        \begin{definition}
            We say that a tableau $\mathscr{T}$ is \emph{completed} if every branch of $\sigma$ is complete or closed.
        \end{definition}

        So, if a branch $\theta$ of a tableau $\mathscr{T}$ is \emph{complete} and \emph{open}, then we have at least one open branch (that is also complete) per subordinated tableaux to $\theta$.

        \begin{theorem}\label{thrm:tableauxModel}
            Every complete and open branch of a tableau is satisfiable.
        \end{theorem}

        \begin{proof}
            Let $\theta$ be a complete and open branch of a tableau $\mathscr{T}$ and $\Sigma$ be a set of formulas of $\theta$ and of the tableaux $\mathscr{T}_1, \mathscr{T}_2, \dots$ (which are recursively subordinated to $\theta$). We construct a model $\mathscr{M}$ where $S$ is the set of tableaux $\{\mathscr{T}, \mathscr{T}_1, \mathscr{T}_2, \dots\}$, $\sim_{a}$ is built from the pairs $(\mathscr{T}_1, \mathscr{T}_2)$, such that $\mathscr{T}_2$ is subordinated to $\mathscr{T}_1$ and satisfying the following conditions, where $e \in E$ is an expression and the prefixes $\sigma, \sigma_1, \sigma_2, \dots$ are associated to $\{\mathscr{T}, \mathscr{T}_1, \mathscr{T}_2, \dots\}$, respectively:

            \begin{enumerate}
                \item if $(\sigma, e) \in \Sigma$, then $V(\sigma, e) = T$;
                \item if $(\sigma, \neg e) \in \Sigma$, then $V(\sigma, e) = F$;
                \item if $(\sigma, e) \not \in \Sigma$ and $(\sigma, \neg e) \not \in \Sigma$, then $V(\sigma, e) = T$ can have any value. Let's choose $F$ by default.
            \end{enumerate}

            Now, for any $(\sigma, \gamma) \in \Sigma$, we have $\mathcal{M},s \Vdash \gamma$, where $\gamma$ is a formula and $s$ a possible state associated to $\sigma$. According to $\gamma$ structure:

            \begin{itemize}
                \item for $(\sigma, p), (\sigma, \varphi), (\sigma, \psi), (\sigma, \langle \pi \rangle \gamma$) and $(\sigma, \neg \langle \pi \rangle \gamma)$ the proof is standard \cite{FITTING1983,GORE1999,MASSACCI1994}. We show the cases for the new rules presented in section \ref{TCDDYLRules} below;

                \item The pair $(\sigma, \{m\}_k) \in \Sigma$, for some prefix $\sigma$. By condition \ref{cond:tableauxCompleteAlpha} of Definition \ref{def:tableauxComplete}, we have $(\sigma, m) \in \Sigma$ and $(\sigma,k) \in \Sigma$ and by the induction hypothesis $\mathcal{M},s \Vdash m$ and $\mathcal{M},s \Vdash k$ and also $\mathcal{M},s \Vdash m \land k$, by the soundness of axiom \ref{ax-ddyl-encryption} of section \ref{ax-ddyl}, we have $\mathcal{M},s \Vdash \{m\}_k$;

                \item The pair $(\sigma, \neg \{m\}_k) \in \Sigma$, for some prefix $\sigma$. By condition \ref{cond:tableauxCompleteBeta} of Definition \ref{def:tableauxComplete}, we have $(\sigma, \neg m) \in \Sigma$ or $(\sigma, \neg k) \in \Sigma$ and by the induction hypothesis $\mathcal{M},s \Vdash \neg m$ or $\mathcal{M},s \Vdash \neg k$ and also $\mathcal{M},s \Vdash \neg m \lor \neg k$ and $\mathcal{M},s \Vdash \neg (m \land  k)$, by the soundness of the contrapositive of axiom \ref{ax-ddyl-encryption} of section \ref{ax-ddyl}, we have $\mathcal{M},s \Vdash \neg \{m\}_k$;

                \item the cases for rules R$_\text{Pair}$ and R$_\text{Pair}^\neg$ are analogous to the cases for rules R$_\text{Dec}$ and R$_\text{Enc}^\neg$, respectively, but using axiom \ref{ax-ddyl-pair} of section \ref{ax-ddyl}.
            \end{itemize}

            Therefore, our model satisfies $\Sigma$.
        \end{proof}

        \begin{theorem}
            If a formula $\gamma$ is valid, then $\gamma$ has a proof by tableaux method.
        \end{theorem}

        \begin{proof}
            Let $\mathscr{T}$ be a completed tableau, started with $\neg \gamma$. If it is open, then $\neg \gamma$ is satisfiable by theorem \ref{thrm:tableauxModel}. So, $\gamma$ cannot be valid. Therefore, if $\gamma$ is valid, then $\mathscr{T}$ is closed and $\gamma$ has a proof by tableaux method.
            \vspace{-5pt}
        \end{proof}

        \subsection{Termination property}
        \label{sec:ddylTableauxTermination}

        Finally, we present the termination argument for our method, based on \cite{MASSACCI2000}.

        \vspace{-10pt}

        \subsubsection{Classical and modal rules}
            \label{subsec:dymaelTableauxTerminationClassicModalRules}

            To guarantee the termination of the proof search, a \emph{loop checking} approach is used, a combination of techniques to apply any rule only after check if it was not applied already to the same antecedent. For the classsical tableaux rules, the following technique is sufficient to terminate:

            \begin{technique}
                \label{tech:dymaelTableauxReduced}

                Apply a rule to a prefixed formula $(\sigma, \gamma)$ in $\theta$ only if the formula is not already reduced according to Definition \ref{def-tableaux-branch-reduced}.
            \end{technique}

            To invocate the \emph{loop checking} we recall the notion of a \emph{fully reduced} prefix, from Definition \ref{def-tableaux-branch-reduced} and the following technique together with Technique \ref{tech:dymaelTableauxReduced} prove that we will always have a $\pi$-completed branch:

            \begin{technique}
                Select the prefixed formulas with the shortest prefix.
            \end{technique}

            As we have a $\pi$-completed branch, the next technique guarantees termination:

            \begin{technique}
                Check if the prefix of a $\pi$-formula is not a copy of a shorter prefix before reducing it.
            \end{technique}

        \subsubsection{Dynamic Dolev-Yao Logic rules}
            As rules R$_\text{Dec}$, R$_\text{Enc}^\neg$, R$_\text{Pair}$, R$_\text{Pair}^\neg$ always yield a smaller conclusion than the premises, that is, they are considered \emph{analytic} rules, the argument explained in Section \ref{subsec:dymaelTableauxTerminationClassicModalRules} is not interfered.

    \begin{example}
        Let's explore a simple case that agent A sends a encrypted message to the intruder Z. Considering that:


        \NumTabs{32}

        {\small

        $Exp(\pi_A \pp \pi_Z) =$ \tab $\tau(m).(\alpha(m_1).END \pp \oa(m_1).END) ~+$\\
        \indent \tab \tab \tab \tab \tab $\oa(m).(\alpha(m_1).END \pp \alpha(m).\oa(m_1).END) ~+$\\
        \indent \tab \tab \tab \tab \tab $\alpha(m).(\oa(m).\alpha(m_1).END \pp \oa(m_1).END)$

        }

        
        \noindent and the following protocol:


        \NumTabs{44}

        {\small

        $\pi_A$ \tab $ = \oa(m).\alpha(m_1).END$\\
        \indent $\pi_Z$ \tab $ = \alpha(m).\oa(m_1).END$\\
        \indent $m$ \tab $ = (A, \{M\}k_Z, Z)$\\
        \indent $m_1$ \tab $ = (Z, \{M\}k_A, A)$

        }


        \noindent we want to know if the intruder get access to the content of such message.

        \

        {\small

        \NumTabs{3}
        \noindent $1. ~ \bar{k_Z}$ \tab \tab [premise]\\
        $2. ~ k_A$ \tab \tab [premise]\\
        $3. ~ \neg \langle \pi_A \pp \pi_Z \rangle M$ \tab \tab [negated conclusion]\\
        $4. ~ \neg \langle \tau(m) ; ( \alpha(m_1) \pp \oa(m_1) ) \rangle M$ \tab \tab [R$_\text{Par}^\neg$ 3]\\
        $5. ~ \neg \langle \tau(m) \rangle \langle \alpha(m_1) \pp \oa(m_1) \rangle M$ \tab \tab [R$_\text{Seq}^\neg$ 4]\\
        $6. ~ \langle \tau(m) \rangle \neg \langle \alpha(m_1) \pp \oa(m_1) \rangle M$ \tab \tab [R$_\text{D}$ 5]\\
        $6.\tau(m).1. ~ m \wedge \neg \langle \alpha(m_1) \pp \oa(m_1) \rangle M$ \tab  \tab[R$_{\langle \tau \rangle}$ 6]\\
        $6.\tau(m).2. ~ m$ \tab \tab [R$_\wedge$ 6.$\tau(m)$.1]\\
        $6.\tau(m).3. ~ \neg \langle \alpha(m_1) \pp \oa(m_1) \rangle M$ \tab \tab [R$_\wedge$ 6.$\tau(m)$.1]\\
        $6.\tau(m).4. ~ \{ M \}_{k_Z}$ \tab \tab [R$_\text{Pair}$ 6.$\tau(m)$.2]\\
        $6.\tau(m).5. ~ M$ \tab \tab [R$_\text{Dec}$ 1, 6.$\tau(m)$.4]\\
        $6.\tau(m).6. ~ \langle \alpha(m_1) \pp \oa(m_1) \rangle \neg M$ \tab \tab [R$_\text{D}$ 6.$\tau(m)$.3]\\
        $6.\tau(m).7. ~ \langle \tau(m_1) ; (END \pp END) \rangle \neg M$ \tab [R$_\text{Par}$ 6.$\tau(m)$.6]\\
        $6.\tau(m).8. ~ \langle \tau(m_1) \rangle \langle END \pp END \rangle \neg M$ \tab  \tab[R$_\text{Seq}$ 6.$\tau(m)$.7]\\
        $6.\tau(m).8.\tau(m_1).1. ~ m_1 \wedge \langle END \pp END \rangle \neg M$ \tab [R$_{\langle \tau \rangle}$ 6.$\tau(m)$.8]\\
        $6.\tau(m).8.\tau(m_1).2. ~ m_1$ \tab \tab [R$_\wedge$ 6.$\tau(m)$.8.$\tau(m_1)$.1]\\
        $6.\tau(m).8.\tau(m_1).3. ~ \langle END \pp END \rangle \neg M$ \tab [R$_\wedge$ 6.$\tau(m)$.8.$\tau(m_1)$.1]\\
        \NumTabs{2}
        Left-hand side:\\
        $6.\tau(m).8.\tau(m_1).3_l.1. ~ \langle END \rangle \neg M$ \tab [R$_\text{Par}$ 6.$\tau(m)$.8.$\tau(m_1)$.3]\\
        $6.\tau(m).8.\tau(m_1).3_l.2. ~ \neg M$ \tab [R$_\text{END}$ 6.$\tau(m)$.8.$\tau(m_1)$.3$_l$.1]\\
        $\bot$ \tab [contradiction with 6.$\tau(m)$.5]\\
        Right-hand side:\\
        $6.\tau(m).8.\tau(m_1).3_r.1. ~ \langle END \rangle \neg M$ \tab [R$_\text{Par}$ 6.$\tau(m)$.8.$\tau(m_1)$.3]\\
        $6.\tau(m).8.\tau(m_1).3_r.2. ~ \neg M$ \tab [R$_\text{END}$ 6.$\tau(m)$.8.$\tau(m_1)$.3$_r$.1]\\
        $\bot$ \tab [contradiction with 6.$\tau(m)$.5]

        }

        Since all the branches are closed, the tableau is closed and $\langle \pi_A \pp \pi_Z \rangle M$.
        \vspace{-3pt}
    \end{example}

\section{Conclusion}\label{conclusion}
    In this work we presented a Dynamic Logic with Parallel operator to verify authenticity and safety in cryptographic protocols. We made this by extending a Dynamic Logic with some concepts based on the Dolev-Yao model. We also provided a tableaux calculus for this logic, proving its termination, soundness and completeness.

    The axiomatization presented here combines axioms from two logics. One is a Propositional Dynamic Logic with Communication Action and Parallel Operator \cite{jolc14} and the other is a Dolev-Yao Multi-agent Epistemic Logic presented in \cite{benevidesetal2018dymael}. As our logic is an extension of the logics presented in \cite{jolc14,tcs/Benevides17} and in \cite{benevidesetal2018dymael}, the proposed tableaux calculus can easily be adapted for these logics.

    As a future work, it would be interesting to extend the tableaux calculus with {\it iteration} operator and with {\it while} (a restrict form of iteration) and {\it deterministic} programs. We also would like to establish the computational complexity of the logics proposed for the model checking and validity problems.

    \subsubsection*{Acknowledgements}
        This study was financed in part by the Coordenação de Aperfeiçoamento de Pessoal de Nível Superior - Brasil (CAPES) - Finance Code 001, by the Brazilian Research Agencies (CNPq) and by the Rio de Janeiro State Research Foundation (FAPERJ).

\bibliographystyle{eptcs}
\bibliography{bibliography}
\end{document}